\documentclass[conference]{IEEEtran}
\usepackage[utf8]{inputenc}
\usepackage{amsmath,amssymb,bm}
\usepackage[dvips]{graphicx}
\usepackage{array}
\usepackage{cite}
\usepackage{enumerate}
\usepackage{color}
\usepackage{float}
\usepackage[capitalise]{cleveref}
\usepackage[mathcal]{euscript}
\usepackage{multirow,multicol,tabu}
\usepackage{subcaption}
\usepackage{soul}

\newcolumntype{?}{!{\vrule width 1pt}}

\usepackage{mathtools}	

\usepackage{algorithm}
\usepackage{algpseudocode}

\newcolumntype{M}[1]{>{\centering\arraybackslash}m{#1}}
\newcolumntype{N}{@{}m{0pt}@{}}

\usepackage{etoolbox}

\makeatletter
\newcommand{\changeoperator}[1]{%
  \csletcs{#1@saved}{#1@}%
  \csdef{#1@}{\changed@operator{#1}}%
}
\newcommand{\changed@operator}[1]{%
  \mathop{%
    \mathchoice{\textstyle\csuse{#1@saved}}
               {\csuse{#1@saved}}
               {\csuse{#1@saved}}
               {\csuse{#1@saved}}%
  }%
}
\makeatother

\changeoperator{sum}
\changeoperator{prod}

\usepackage{tikz,tabularx}
\usetikzlibrary{decorations.pathreplacing}
\tikzstyle{every picture}+=[remember picture]

\usetikzlibrary{shapes,arrows,shadows}	
\usepgflibrary{patterns}
\usetikzlibrary{patterns}
\usetikzlibrary{shapes.geometric}
\usetikzlibrary{arrows}

\usetikzlibrary{decorations.markings}

\tikzset{myptr1/.style={decoration={markings,mark=at position 1 with %
    {\arrow[scale=2.5]{>}}},postaction={decorate}}}
    
\tikzset{myptr2/.style={decoration={markings,mark=at position 1 with %
    {\arrow[scale=1.8]{>}}},postaction={decorate}}}
    
\usetikzlibrary{decorations.markings}

\usepackage{multicol}
\usepackage{amsthm}

\newtheorem{definition}{Definition}
\newtheorem{theorem}{Theorem}

\newtheorem{proposition}{Proposition}

\def \sc {\scriptsize}

\newcommand\blfootnote[1]{%
  \begingroup
  \renewcommand\thefootnote{}\footnote{#1}%
  \addtocounter{footnote}{-1}%
  \endgroup
}

\title{Covert queueing problem with a Markovian statistic}

\author{\IEEEauthorblockN{Arti Yardi}
 	\IEEEauthorblockA{IIIT Bangalore}
 	\IEEEauthorblockA{arti.yardi@iiitb.ac.in}		
 	\and
 	\IEEEauthorblockN{Tejas Bodas\textsuperscript{$\dagger$}}
 	\IEEEauthorblockA{TCS Research, India}
 	\IEEEauthorblockA{tejas.bodas@tcs.com}		
}

\begin{document}

\maketitle

%==================================
\begin{abstract}
Based on the covert communication framework, we consider a covert queueing problem that has a Markovian statistic. Willie jobs arrive according to a Poisson process and require service from server Bob. Bob does not have a  queue for jobs to wait and hence when the server is busy, arriving Willie jobs are lost. Willie and Bob enter a contract under which Bob should only serve Willie jobs. As part of the usage statistic, for a sequence of N consecutive jobs that arrived, Bob informs Willie whether each job was served or lost (this is the Markovian statistic). Bob is assumed to be violating the contract and admitting non-Willie (Nillie) jobs according to a Poisson process. For such a setting, we identify the hypothesis testing to be performed (given the Markovian data) by Willie to detect the presence or absence of Nillie jobs. We also characterize the upper bound on arrival rate of Nillie jobs such that the error in the hypothesis testing of Willie is arbitrarily large, ensuring covertness in admitting Nillie jobs. 
\end{abstract}
%==================================


%==================================
%==================================
% 
\section{Introduction}
\label{Section_Introduction}

%--------------------------------------
% 
\blfootnote{$\dagger$This work was done when the author was a faculty at IIT Dharwad.}

In the problem of \textit{covert communication}, one considers a setup where Alice is transmitting messages to Bob and intruder Willie is snooping over this communication. The aim of Willie is to determine if the communication between Alice and Bob is taking place or not~\cite{Bash_sq_root_law_2012, Bash_sq_root_journal_2013}. While this problem has been studied for a wide variety of system 
models~\cite{Towsley_poisson_channel_2015,
Towsley_Wireless_2017, Towsley_Change_Point_2017, Towsley_packet_insertion_2020,Bash16, 
Jaggi_2013, Jaggi_2014,Jaggi_channel_uncertainty_2014, Jaggi_poly_complex_2020, 
Bloch_ITW_2016, Bloch_resolvability_2016, Kadampot20, Bloch_Rayleigh_2020, Bloch_Broadcast_2019}, in this paper we consider a novel setup of \textit{covert queueing} problem first studied in \cite{Towsley_cycle_stealing_2020}.
In a covert queueing problem, there are three entities, namely Willie, Bob, and Nillie (for non-Willie). Bob is a server that processes incoming jobs that arrive to its queue. 
While Bob is obligated to serve only Willie jobs, he may  allow some non-Willie or Nillie jobs (for his selfish motives).
Willie wants to determine whether Bob is allowing such illegitimate Nillie traffic or not, while Bob wants to have as much Nillie traffic as possible without Willie being able to detect its presence, hence the name \textit{covert queueing}.
Similar to the problem of covert communication, the aim in the covert queueing problem is to determine the asymptotic limit on the arrival rate of Nillie jobs such that Willie is not able to detect the presence of Nillie traffic with high probability.
%(see \cref{Definition_ep_covertness}).

This covert queueing problem was first introduced in 
\cite{Towsley_cycle_stealing_2020} and a recent variant appeared in \cite{Comm_letter_2021}.
 In both these models, Willie jobs arrive according the a Poisson process and Willie and Bob enter into an exclusive contract under which Bob should only serve Willie jobs.
%  Bob however wants to admit non-Willie or Alice or Nillie jobs covertly without Willie knowing about it. In doing so, Bob is able to generate extra revenue in serving Nillie jobs, but this is at the cost of Willie jobs having to now contend with Nillie jobs for access to the server.
%  
 As proof of exclusive work, Bob periodically informs Willie of some usage statistic. The statistic that Bob uses is typically opaque so that Willie is not able to detect Nillie jobs in an obvious manner.
For example, in \cite{Towsley_cycle_stealing_2020}, only the arrival and departure times of Willie jobs are used as the information statistic. Similarly in \cite{Comm_letter_2021}, either the length of successive $N$ busy periods or the number of jobs served in these busy periods is conveyed. Due to the opaqueness of such statistics, Willie must perform hypothesis testing to infer the presence or absence of Nillie jobs. In turn, Bob desires to 
%knows the kind of hypothesis testing that Willie can perform and 
admit Nillie jobs in such a way that the probability of error in the hypothesis testing is close to one.
The strategy for admitting Nillie jobs is dependent on the information statistic conveyed. For example, \cite{Towsley_cycle_stealing_2020} employs the probabilistic insert at the end of busy period (IEBP) strategy
%given the arrival and departure times information. In
while in \cite{Comm_letter_2021}, Nillie jobs are admitted according to a Poisson process with rate $\lambda_b$ and the upper bound of this rate that maintains covertness is determined. 

A common feature of both the works described above is that the sequence of information statistic conveyed by Bob are independent and identically distributed 
(i.i.d.)~random variables. In \cite{Towsley_cycle_stealing_2020}, the arrival and departure time information is used to generate a sequence of reconstructed service times that are i.i.d. Similarly, in \cite{Comm_letter_2021}, the statistic includes successive busy period information which is again i.i.d.
This i.i.d.~nature of the statistic makes the hypothesis testing problem more amenable to analysis. 

In this work, we depart from  i.i.d.~statistics and instead consider a Markovian statistic. We assume that the server Bob does not have a buffer for arriving jobs to wait and the jobs that find server to be busy are lost. (Analysis for case with buffer is for future work.) 
We assume that Willie and Nillie jobs arrive according to a Poisson process with rates $\lambda_w$ and $\lambda_b$ respectively and their service times are i.i.d.~with exponential distribution of parameter $\mu$ where $\mu > \lambda_w + \lambda_b$. As part of usage statistics, Bob provides Willie with a sequence of $N$ random variables $X_1, \ldots, X_N$ (associated with $N$ successive arrivals), where $X_j = 1$ if the $j$-th arrival finds the server to be busy and is lost and $X_j = 0$ if the server is idle. Here the $N$ jobs may include Nillie jobs as well and since the arrival or departure time information is not conveyed, it is difficult for Willie to determine presence or absence of Nillie jobs by observing $X_1, \ldots, X_N$. This statistic is Markovian and therefore the hypothesis testing as performed by Willie is based on detection of Markov chains with known parameters \cite[Ch.~12]{Levy2008}. 
As part of our main result, we formulate this 
covert queueing setup as a 
hypothesis testing problem with the Markovian statistic and identify the upper bound on $\lambda_b$ that ensures covertness. 
We first obtain a closed form expression for the error exponent associated with this hypothesis testing (see \cref{Proposition_Error_exponent}).
Using this we then show that the arrival rate $\lambda_b$ of Nillie jobs should be of order $\mathcal{O}(\sqrt{\log(K(N))/N})$, where $K(N)$ is a function of $N$ that decreases with $N$ at a sub-exponential rate (see \cref{Eqn_const_KF_KM}, \cref{Theorem_asymptotic_result}).

The rest of the paper is organized as follows.
In \cref{Section_system_model}, we first provide details of our problem setup and then recall some of the basics of hypothesis testing between two Markov chains (for the sake of completeness).
The main results of the paper are discussed in \cref{Section_results}, followed by some concluding remarks in \cref{Section_Conclusion}.

%==================================
%==================================
% 
\section{System model}
\label{Section_system_model}

Bob is a server that Willie has contracted to serve its jobs. We assume that Willie and Nillie jobs arrive to Bob according to a Poisson process with rates $\lambda_w$ and $\lambda_b$ respectively. 
The server time for each job is exponentially distributed with rate $\mu$.
Bob does not have a queue or buffer to store arriving jobs in which case  arriving jobs are lost when the server is busy serving a previous arrival. Consider a sequence of $N$ successive arrivals and let $X_{j}$ denote the state of the server (whether it is busy or idle) as seen by the $j$-th job. More precisely, $X_{j} = 1$ if the arrival found the server to be busy in which case the job is lost. When $X_{j} = 0$, the arriving job finds the server to be idle and starts receiving service. 

As part of usage statistics, Bob provides Willie with the sequence of random variables $X_1^N \coloneqq \{X_1, \ldots, X_N\}$, where the choice of $N$ is arbitrary.
Note that Bob does not tell Willie any information about the arrival and departure time of each job.  
Since the arrivals and Poisson and the service times are exponential, the system can be represented by an $M/M/1/1$ queue and the sequence of random variables $X_1^N$ constitute a Markov chain~\cite{Queueing_systems_book}. 
Let $x_j$ be the realization of $X_{j}$ for $j = 1, 2, \ldots, N$.
Given a sequence $x_1^N$, the aim of Willie is to determine if Bob is inserting any non-Willie jobs or not using a binary hypothesis testing problem given by
\begin{equation}
\begin{aligned}
H_0&: \mbox{Bob is serving only Willie jobs}\\
H_1&: \mbox{Bob is serving Willie as well as Nillie jobs.}
\end{aligned}
\label{Eqn_hypothesis_testing}
\end{equation}

Since the underlying system is an $M/M/1/1$ queue, one can see that the value $X_j$ for the $j$-th job only depends on $X_{j-1}$ but not on any of the variables $X_k$ for $k \leq j-2$~\cite{Queueing_systems_book}. 
% sis a function of server service rate $\mu$, arrival rates $\lambda_w$, and \tcb{$\lambda_n$} of Willie and non-Willie jobs respectively and the number of jobs already waiting in the buffer. 
% 
Since $X_{j}$ denotes the state of the server as seen by the $j$-th arrival and the arrival process is Poisson, $\{X_j, j \geq 0\}$ forms a two state discrete time Markov chain. $x_1^N$ denotes an N length realization or sample path of this Markov chain and the transition probability of this Markov chain is a function of the parameters $\lambda_w, \lambda_b,$ and $\mu$.
% depends only on $X_{{j}-1}$ and not on $X_{1}^{{j}-2}$, the sequence
%$$X_1^N$ form an \tcr{order one} Markov chain. 
% 
Let $P$ and $Q$ denote the state transition matrices under hypotheses $H_0$ and $H_1$ respectively. 

We assume that both the hypotheses are equally likely. 
Let $P_F(N)$ be the probability of rejecting $H_0$ under the condition it is true and $P_M(N)$ is the probability of accepting $H_0$ under the condition $H_1$ is true.
For the equally likely hypotheses, the total error $P_E(N)$ is equal to $P_E(N) = (P_F(N) + P_M(N))/2$~\cite{Lehmann_Book}.
In order to achieve covertness, note that Bob wishes to have $P_E(N)$ close to one. The covertness criteria is defined formally in \cref{subsection_Covertness_criterion}.

%-------------------------------------------------------
\subsection{Preliminaries about hypothesis testing between two Markov chains~\cite[Ch.~12]{Levy2008}}

We shall now summarize some basics about hypothesis testing between two discrete-time, finite-state Markov chains with state transition matrices $P$ and $Q$. 
% 
% \tcb{We also assume that both the Markov chains are irreducible and aperiodic.}
% 
Details can be found in \cite[Ch.~12]{Levy2008} and references therein. 
% 
% Under hypothesis $H_0$, the given sequence $x_1^N$ corresponds to state
% 
For the given sequence $x_1^N$, we wish to determine whether it corresponds to Markov chain with state transition matrix $P$ (hypothesis $H_0$) or $Q$ (hypothesis $H_1$).
Let $\mathbb{P}[x_1^N|H_j]$ be the probability of observing $x_1^N$ under the condition that hypothesis $H_j$ is true, for $j = 0,1$
and consider the log-likelihood ratio $\log L(x_1^N)$\footnote{We assume natural logarithm throughout the paper.} given by
\begin{align}
\log L(x_1^N) \coloneqq \log \frac{\mathbb{P}[x_1^N|H_0]}{\mathbb{P}[x_1^N|H_1]}.
\end{align}
The hypothesis testing consists of determining a threshold $\gamma$ such that the decision $H_0$ is true if $\log L(x_1^N) \geq \gamma$ and $H_1$ is true otherwise. This is represented by 
\begin{align}
\log L(x_1^N) \overset{H_0}{\underset{H_1}{\gtreqless}} \gamma.
\end{align}
In the asymptotic setting, when $N$ tends to infinity, one can apply the G\"{a}rtner-Ellis theorem to obtain the asymptotic values of $P_F(N)$ and $P_M(N)$ which is given by (see Eq. (12.50) of \cite{Levy2008})
\begin{equation}
\begin{aligned}
\lim_{N \rightarrow \infty}\frac{1}{N} \log P_F(N) &= -I_0(\gamma) \\ 
\lim_{N \rightarrow \infty}\frac{1}{N} \log P_M(N) &= -I_1(\gamma) = \gamma -I_0(\gamma), 
\end{aligned}
\label{Eqn_asymptotic_PF_PM}
\end{equation}
where $I_0(\gamma)$ is termed as the \textit{error exponent}. It is also called as \textit{rate of decay error} since it is the rate with which the error in the hypothesis testing goes to zero as $N$ tends to $\infty$.
% 
% \tcb{If the following is a theorem from levy book, we should directly paste it here. Otherwise it is not clear where the following details have appeared from. }
% 
Let $K_F(N)$ and $K_M(N)$ be some functions of $N$ that go to zero as $N$ tends to infinity at a sub-exponential rate, i.e., 
\begin{equation}
\begin{aligned}
\lim_{N \rightarrow \infty}\frac{1}{N} \log K_F(N) = 0 \mbox{ \& }
\lim_{N \rightarrow \infty}\frac{1}{N} \log K_M(N) = 0.
\end{aligned}
\label{Eqn_const_KF_KM}
\end{equation}
From \cref{Eqn_asymptotic_PF_PM} and \cref{Eqn_const_KF_KM}, $P_F(N)$ and $P_M(N)$ can be approximated as~(see Eq. (12.53) of \cite{Levy2008})
\begin{align}
P_F(N) & \approx K_F(N) \exp(-I_0(\gamma) N )\\
P_M(N) & \approx K_M(N) \exp(-I_1(\gamma) N ).
\end{align}
For our problem since both the hypotheses are assumed to be equally likely, we choose the threshold $\gamma$ that maximizes the rate of decay of the total error $P_E(N) = (P_F(N) + P_M(N))/2$ and this can be achieved by setting $I_0(\gamma) = I_1(\gamma)$. From \cref{Eqn_asymptotic_PF_PM}, this implies that $\gamma = 0$.
In this case, the functions $K_F(N)$ and $K_M(N)$ can also be chosen to be equal.
Suppose $K_F(N) = K_M(N) = K(N)$ and $I_0(\gamma) = I_1(\gamma) = I_{err}$ and using this $P_E(N)$ can be approximated as
\begin{align}
P_E(N) \approx K(N) \exp(-I_{err} N ).
\label{Eqn_Error_exponent_total_error}
\end{align}
Having outlined the expression for the error in the hypothesis testing, we  now define the covertness criteria for our problem. See \cite{Bash_sq_root_law_2012, Towsley_poisson_channel_2015, Jaggi_2013,Comm_letter_2021} for a similar definition.

%-------------------------------------------------------
\subsection{Covertness criterion}
\label{subsection_Covertness_criterion}
% 
% For the given constant $\epsilon$ with $0 < \epsilon < 1$, the covertness criteria is defined below. Note that this definition of covertness is typically used in the literature~\cite{Bash_sq_root_law_2012, Towsley_poisson_channel_2015, Jaggi_2013}.
% 
\begin{definition}
\label{Definition_ep_covertness}
We say that server Bob is able to insert Nillie jobs with $\epsilon$-covertness if the error in the hypothesis testing $P_E(N)$ satisfies $P_E(N) \geq 1 - \epsilon$ where $0 < \epsilon < 1$.
\hfill $\square$
\end{definition}
Note that while the covertness criteria can be defined for arbitrarily chosen values of $\epsilon$, typically in the literature $\epsilon$ is chosen to be close to zero.
The results provided in our work are however applicable for arbitrary choice of $\epsilon$. 
% For example, $0.4$-covertness implies that Bob wishes to insert non-Willie jobs at rate $\lambda_b$ such that the probability of error in the hypothesis testing as performed by Willie is at least $0.6$.
% , i.e., $\alpha + \beta \geq 0.6$.
% 
% Similarly for $0$-covertness, Bob wishes to insert non-Willie jobs at a rate $\lambda_b$ such that $\alpha + \beta =1$.
% As part of our main result, we show that the largest $\lambda_b$ for $\epsilon$-covertness is directly proportional to $\epsilon$. While Bob can admit non-Willie jobs at a higher rate for higher values of $\epsilon$, this will also enable Willie to detect the presence of non-Willie jobs with a higher probability.
In this paper, we are interested in the asymptotic regime, where $N$ tends to infinity, and hence $P_E(N)$ is close to $K(N) \exp(-I_{err} N )$ (see \cref{Eqn_Error_exponent_total_error}). Thus for our problem, for $\epsilon$-covertness we need
\begin{align}
K(N) \exp(-I_{err} N ) \geq 1 - \epsilon.
\label{Eqn_covertness_criterion_asymptotic}
\end{align}

\vspace{0.03in}

%==================================
%==================================
% 
\section{Main results}
\label{Section_results}

In order to characterize the error $P_E(N)$ in the hypothesis testing of \cref{Eqn_hypothesis_testing}, the key step is to obtain an expression for the error exponent $I_{err}$ (see \cref{Eqn_Error_exponent_total_error}).
In this section, we first provide a closed form expression for this error exponent in \cref{Proposition_Error_exponent}. Using this we then study the asymptotic performance of the hypothesis testing, satisfying the given $\epsilon$-covertness in  
\cref{subsection_asymptotic}.

%==================================
\subsection{Characterizing the error exponent}
\label{subsection_err_exponent}

Recall that in our setup, the   usage statistic provided by Bob is a sequence of random variables $X_1^N$, where each $X_j$ can either be zero or one. Further, $X_1^N$ form a two state Markov chain either with state transition matrix $P$ (hypothesis $H_0$) or $Q$ (hypothesis $H_1$).
The error exponent associated with these two Markov chain hypothesis testing is characterized in the following proposition. 

\begin{proposition}
\label{Proposition_Error_exponent}
When the bufferless server Bob provides Willie with the sequence of server states for $N$ successive arrivals, the error exponent $I_{err}$ (see \cref{Eqn_Error_exponent_total_error}) for the hypothesis testing of \cref{Eqn_hypothesis_testing} is given by
\begin{align*}
I_{err} &= -\log \Big( p^v q^{1-v} + (1-p)^v (1-q)^{1-v} \Big)  \mbox{ with}\\
v &= \log \left( a \frac{\log bc}{\log (1/c)}\right)\Big/\log b,
\end{align*}
where $p = \mu/(\lambda_w + \mu),q = \mu/(\lambda_b + \lambda_w+\mu), a = (\lambda_b + \lambda_w)/\mu, b = (\lambda_b + \lambda_w)/\lambda_w$ and $c = (\lambda_w + \mu)/(\lambda_b + \lambda_w + \mu)$.
\end{proposition}
\begin{proof}
We first obtain the state transition matrix $P$ under hypothesis $H_0$. 
Note that the state of the Markov chain (an arrival finds the server either busy or idle), can either be zero or one.
Suppose both $(j-1)$-th and $j$-th arrivals find the server states to be idle, i.e., $X_{j-1}=0$ and $X_j=0$.
This happens when the service time for the $(j-1)$-th job is less than the interarrival time between the two jobs.
Since the inter-arrival times and service times are  exponential random variables with parameters $\mu$ and $\lambda_w$ respectively, we have $\mathbb{P}[X_j=0|X_{j-1}=0] = \mu/(\lambda_w+\mu)$. 
The remaining state transition probabilities can be calculated in a similar fashion and matrix $P$ is given by
\begin{equation}
\begin{aligned}
P = 
\begin{bmatrix}
    \frac{\mu}{\lambda_w + \mu} & \frac{\lambda_w}{\lambda_w + \mu}\\
    \frac{\mu}{\lambda_w + \mu} & \frac{\lambda_w}{\lambda_w + \mu}
\end{bmatrix}
= 
\begin{bmatrix}
    p & 1-p\\
    p & 1-p
\end{bmatrix}.
\end{aligned}
\label{Eqn_matrix_P}
\end{equation}
Under hypothesis $H_1$, the arrival rate of the jobs is 
$\lambda_b + \lambda_w$ and hence the state transition matrix $Q$ is given by
\begin{equation}
\begin{aligned}
Q = 
\begin{bmatrix}
    \frac{\mu}{\lambda_w + \lambda_b + \mu} & \frac{\lambda_w + \lambda_b}{\lambda_w + \lambda_b + \mu}\\
    \frac{\mu}{\lambda_w + \lambda_b + \mu} & \frac{\lambda_w + \lambda_b}{\lambda_w + \lambda_b + \mu}
\end{bmatrix}
= 
\begin{bmatrix}
    q & 1-q\\
    q & 1-q
\end{bmatrix}.
\end{aligned}
\label{Eqn_matrix_Q}
\end{equation}

Corresponding to matrices $P$ and $Q$ and for a constant $u \in [0,1]$ we now define a matrix $M(u)$ as follows 
% Hadamard product~\cite[Ch.~5]{Hadamard_MC_Book} 
% of size $2 \times 2$ such that the entry in the $i$th row and $j$th column of $M(u)$ is equal to $p_{ij}^u q_{ij}^{1-u}$.
% 
% For the matrices $P$ and $Q$ given in \cref{Eqn_matrix_P} and \cref{Eqn_matrix_Q} respectively, $M(u)$ will be
% 
\begin{align}
M(u) = 
\begin{bmatrix}
p^uq^{1-u} & (1-p)^u (1-q)^{1-u}\\
p^uq^{1-u} & (1-p)^u (1-q)^{1-u}
\end{bmatrix}.
\end{align}
Let $r(u)$ be the spectral radius of $M(u)$. For the hypothesis testing between Markov chains with state transition matrices $P$ and $Q$, the error exponent of \cref{Eqn_Error_exponent_total_error} is 
given by~\cite[Sec.12.2.3]{Levy2008}
\begin{align}
I_{err} = - \min_{0 \leq u \leq 1} \log r(u).
\label{Eqn_err_expo_thm_Levy}
\end{align}

To find $I_{err}$ we thus need to first find the spectral radius $r(u)$ of $M(u)$. Since $p,q \geq 0$, $M(u)$ is a positive square matrix~\cite[Definition~2.1]{Schaefer1974} and from Proposition~2.4 of 
\cite{Schaefer1974}, $r(u)$ is upper and lower bounded by the maximum and minimum row sum of $M(u)$.
Observe that all rows of $M(u)$ are the same and hence $r(u) = p^u q^{1-u} + (1-p)^u (1-q)^{1-u}$. 
We now substitute the values of $p$ and $q$ in $r(u)$ to obtain
\begin{equation}
\begin{aligned}
r(u) 
% 
% &= \left( \frac{\mu}{\lambda_w + \mu} \right)^u \left( \frac{\mu}{\lambda_b + \lambda_w + \mu} \right)^{1-u} 
% + \\
% &\mbox{~~~~~~~~~~~~~~~}\left( \frac{\lambda_w}{\lambda_w + \mu} \right)^u \left( \frac{\lambda_b + \lambda_w}{\lambda_b + \lambda_w + \mu} \right)^{1-u} \\
% 
&= \left( \frac{\mu}{\lambda_b + \lambda_w + \mu} \right) \left( \frac{\lambda_b + \lambda_w + \mu}{\lambda_w + \mu} \right)^{u}
+ \\
&\mbox{~~~~~~~}
\left( \frac{\lambda_b + \lambda_w}{\lambda_b + \lambda_w + \mu} \right) 
\left(
\frac{\lambda_w}{\lambda_b + \lambda_w}
\frac{\lambda_b + \lambda_w + \mu}{\lambda_w + \mu} \right)^{u}.
\\
% &\coloneqq AB^u + CD^u,
% 
\end{aligned}
\label{Eqn_ru_substi}
\end{equation}
Suppose
$A \coloneqq \mu/(\lambda_b + \lambda_w + \mu), 
B \coloneqq (\lambda_b + \lambda_w + \mu)/(\lambda_w + \mu), 
C \coloneqq (\lambda_b + \lambda_w)/(\lambda_b + \lambda_w + \mu)$ and
$D \coloneqq \lambda_w(\lambda_b + \lambda_w + \mu)/(\lambda_b + \lambda_w)(\lambda_w + \mu)$. Substituting the values of $A,B, C,$ and $D$ in \cref{Eqn_ru_substi} we get
\begin{align}
r(u) = AB^u + CD^u.
\label{Eqn_ru_ABCD}
\end{align}

To complete the proof, from \cref{Eqn_err_expo_thm_Levy} we now need to find the value of $u \in [0,1]$ that minimizes $\log r(u)$.
Let $v$ be this minimizer. Since log is concave function, $v$ is given by the solution of $\frac{\mathrm{d}}{\mathrm{d}u} r(u) = 0$ and hence using \cref{Eqn_ru_ABCD}, the minimizer $v$ should satisfy the following equation 
\begin{align}
AB^v \log B + CD^v \log D = 0 \\
\Rightarrow \left(\frac{B}{D} \right)^v = \frac{C \log (1/D)}{A \log B}\\
\Rightarrow v = \frac{\log \left( \frac{C}{A} \frac{\log(1/D)}{\log(B)} \right)}{\log (B/D)} \label{Eqn_v}.
\end{align}
The required expression of $v$ is obtained by  simplifying \cref{Eqn_v} and this completes the proof.
\end{proof}
% 

%==================================
\subsection{Asymptotic performance of hypothesis testing}
\label{subsection_asymptotic}

For ease of exposition and without loss of generality, we will assume for the remainder of this paper that the service rate $\mu$ is set to $1$. We now have the following main theorem that characterizes the asymptotic performance of the underlying hypothesis testing problem, while ensuring $\epsilon$-covertness.
\begin{theorem}
\label{Theorem_asymptotic_result}
%  
%Consider the parameters $ I_{err}, v, p,q,a,b,$ and $c$ as defined in \cref{Proposition_Error_exponent}.
% 
Suppose as the usage statistic, server Bob provides Willie with a sequence of server states for $N$ successive arrivals. 
Then Bob can insert Nillie packets covertly if
$\lambda_b \leq \sqrt { \frac{8\lambda_w (\lambda_w + 1)^2}{N} \log \frac{K(N)}{1 - \epsilon} }$, 
where $K(N)$ is a function of $N$ that decreases with $N$ at a sub-exponential rate (see \cref{Eqn_const_KF_KM}).
Further, $\lambda_b$ should be of order $\mathcal{O}(\sqrt{\log(K(N))/N})$.
\end{theorem}
\begin{proof}
For $\epsilon$-covertness we need (see \cref{Eqn_covertness_criterion_asymptotic})
\begin{align}
K(N) \exp(-I_{err} N ) \geq 1 - \epsilon, 
\label{Eqn_covertness_criterion_asymptotic_recall}
\end{align}
where the expression of $I_{err}$ is provided in \cref{Proposition_Error_exponent}. We observe that for the asymptotic analysis, the expression of $I_{err}$ provided in \cref{Proposition_Error_exponent} is not amenable for further analysis.
In the literature, typically such situations are resolved with the aid of Taylor series approximation (for example see \cite{Bash_sq_root_law_2012, Towsley_poisson_channel_2015, Towsley_cycle_stealing_2020} and references therein). 
Along similar lines, we consider the second order Taylor series approximation of $I_{err}$
with respect to $\lambda_b$ around $\lambda_b = 0$. 
We denote $I_{err}$ by $I_{err}(\lambda_b)$ to indicate that it is a function of $\lambda_b$ (since $\lambda_w$ is now treated as a constant).
% 
%As explained in \cref{Remark_Tayler_closeness}, 
We have observed that, 
%in the asymptotic regime (when $N \rightarrow$), $\lambda_b$ is close to zero (this can be observed from \cref{Eqn_UB_lambda_b_main}) and 
the Taylor series expansion of $I_{err}(\lambda_b)$ is very close to its true value in the asymptotic regime (when $N \rightarrow \infty$).

The remainder of this proof is now sub-divided into two parts.
We first find the Taylor series expansion of $I_{err}(\lambda_b)$ and then use it to obtain required asymptotic result for satisfying $\epsilon$-covertness criteria. 
In what follows, all the derivatives are taken with respect to $\lambda_b$ and $f^{\prime}(\lambda_b)$ denotes the derivative of a function $f(\lambda_b)$.
% 

%---------------------------------------------
~\\
\textit{Part-I: Obtaining Taylor series expansion of $I_{err}(\lambda_b)$}

Note that the calculations towards finding the Taylor series expansion are not straightforward (not amenable via MATLAB/Mathematica) and hence we provide these steps in detail. 
Suppose $I_{err}(\lambda_b)$ is approximated via
the second order Taylor series approximation around $\lambda_b = 0$ as follows
\begin{align}
I_{err}(\lambda_b) \approx I_{err}(0) + I_{err}^{\prime}(0) \lambda_b + \frac{I_{err}^{\prime \prime}(0)}{2} \lambda_b^2,
\label{Eqn_Taylor_series_approx}
\end{align}
where $I_{err}^{\prime}(0)$ and $I_{err}^{\prime \prime}(0)$ are the first and second derivative of $I_{err}(\lambda_b)$, evaluated at $\lambda_b = 0$.
We have observed numerically that considering terms upto $\lambda_b^2$ provide a good approximation.
In the expression of $I_{err}(\lambda_b)$, observe that while $q$ and $v$ are functions of $\lambda_b$, $p$ is not. To indicate this dependence explicitly we denote $q$ and $v$ by $q(\lambda_b)$ and $v(\lambda_b)$ respectively. 
Suppose $\bar{p} = 1-p, \bar{q}(x) = 1-q(x), F_1(\lambda_b) = p^{v(\lambda_b)} 
q(\lambda_b)^{1-v(\lambda_b)}, F_2(\lambda_b) = \bar{p}^{v(\lambda_b)} 
\bar{q}(\lambda_b)^{1-v(\lambda_b)}$, and $F(\lambda_b) = F_1(\lambda_b) + F_2(\lambda_b)$. 
Using this we have $I_{err}(\lambda_b) = -\log F(\lambda_b)$.
First and second derivatives of $I_{err}(\lambda_b)$ are now given by
\begin{align}
I_{err}^{\prime}(\lambda_b) = -\frac{F^{\prime}(\lambda_b)}{F(\lambda_b)},
I_{err}^{\prime \prime}(\lambda_b) = \frac{(F^{\prime}(\lambda_b))^2 - F(\lambda_b) F^{\prime \prime}(\lambda_b)  }{(F(\lambda_b))^2}.
\label{Eqn_I_err_derivatives}
\end{align}
To obtain the Taylor series expansion, we thus need to 
evaluate $F(\lambda_b), F^{\prime}(\lambda_b)$ and $F^{\prime \prime}(\lambda_b)$ as $\lambda_b \rightarrow 0$. Towards this we first note that $\lim_{\lambda_b \rightarrow 0} v(\lambda_b) = 1/2$ (the proof for this involves repeated application of L'Hopital's rule and we skip this due to space constraints).
% % 
% Since this limit can be easily evaluated using MATLAB, due to space constraints, we skip the details of proof. 
% % 
Further, it can be easily verified that $q(0)=p, q^{\prime}(0) = -p^2, q^{\prime \prime}(0) = 2p^3, \bar{q}(0)=1-p, \bar{q}^{\prime}(0) = p^2$, and $\bar{q}^{\prime \prime}(0) = -2p^3$. We now evaluate $F(\lambda_b), F^{\prime}(\lambda_b)$, and $F^{\prime \prime}(\lambda_b)$ as $\lambda_b \rightarrow 0$.

~\\
%------------------------------------------------ 
(1) \textit{Evaluating $\lim_{\lambda_b \rightarrow 0}F(\lambda_b)$}:
$F(\lambda_b)$ is given by
\begin{align}
F(\lambda_b) = p^{v(\lambda_b)} q(\lambda_b)^{1-v(\lambda_b)} + \bar{p}^{v(\lambda_b)} 
\bar{q}(\lambda_b)^{1-v(\lambda_b)}.
\label{Eqn_F_lambda_b}
\end{align}
For any arbitrary functions $f(x)$ and $g(x)$, it is known that $\lim_{x \rightarrow x_0} f(x)^{g(x)} = f_0^{g_0}$, where $\lim_{x \rightarrow x_0} g(x) = g_0$ and $\lim_{x \rightarrow x_0} f(x) = f_0$.
It can be seen that $\lim_{\lambda_b \rightarrow 0}q(\lambda_b) = p$ and since $\lim_{\lambda_b \rightarrow 0} v(\lambda_b) = 1/2$ in \cref{Eqn_F_lambda_b} we get
\begin{align}
\lim_{\lambda_b \rightarrow 0} F(\lambda_b) = p^{0.5} p^{(1-0.5)} + \bar{p}^{0.5} 
\bar{p}^{(1-0.5)} = 1.
\label{Eqn_F_lambda_b_limit}
\end{align}

~\\ 
%------------------------------------------------ 
(2) \textit{Evaluating $\lim_{\lambda_b \rightarrow 0}F^{\prime}(\lambda_b)$}: Let us first find $F_1^{\prime}(\lambda_b)$. With some simple calculations it can be shown that
\begin{align}
F_1^{\prime}(\lambda_b) &= 
p^{v(\lambda_b)} q(\lambda_b)^{1-v(\lambda_b)} 
\bigg[ v^{\prime}(\lambda_b) \Big(\log p - \log q(x) \Big)
\nonumber
\\
& \mbox{~~~~~~~~~~~~~~~~~~~~~~~~~~} + \Big(1-v(\lambda_b)\Big)\frac{q^{\prime}(\lambda_b)}{q(\lambda_b)} \bigg].
\nonumber
\\
&\coloneqq F_1(\lambda_b) T_1(\lambda_b),
\label{Eqn_F1_prime_lambda_b}
\end{align}
where $T_1(\lambda_b) \coloneqq v^{\prime}(\lambda_b) (\log p - \log q(x) )
+ (1-v(\lambda_b))\frac{q^{\prime}(\lambda_b)}{q(\lambda_b)}$.
Substituting $q(0) = p, q^{\prime}(0) = -p^2,$ and $\lim_{\lambda_b \rightarrow 0} v(\lambda_b) = 1/2$ in \cref{Eqn_F1_prime_lambda_b}, it can be verified that $\lim_{\lambda_b \rightarrow 0}F_1(\lambda_b) = p,\lim_{\lambda_b \rightarrow 0}T_1(\lambda_b) = -p/2$ and hence $\lim_{\lambda_b \rightarrow 0}F_1^{\prime}(\lambda_b) = -p^2/2$.
Using similar steps it can be shown that $\lim_{\lambda_b \rightarrow 0}F_2^{\prime}(\lambda_b) = p^2/2$ and hence
\begin{align}
\lim_{\lambda_b \rightarrow 0}F^{\prime}(\lambda_b) = 
\lim_{\lambda_b \rightarrow 0}F_1^{\prime}(\lambda_b) + F_2^{\prime}(\lambda_b) = \frac{-p^2}{2} + \frac{p^2}{2} = 0.
\label{Eqn_F_prime_limit}
\end{align}

~\\
%------------------------------------------------
(3) \textit{Evaluating $\lim_{\lambda_b \rightarrow 0}F^{\prime \prime}(\lambda_b)$}: 
Let us first find $F_1^{\prime \prime}(\lambda_b)$. From \cref{Eqn_F1_prime_lambda_b} we have,
\begin{align}
F_1^{\prime \prime}(\lambda_b) = F_1(\lambda_b)T_1^{\prime}(\lambda_b) + F_1^{\prime}(\lambda_b)T_1(\lambda_b)
\label{Eqn_F1_prime_limit}
\end{align}
where, with some calculations $T_1^{\prime}(\lambda_b)$ is obtained as
\begin{align}
T_1^{\prime}(\lambda_b) 
& = -2v^{\prime}(\lambda_b) \frac{q^{\prime}(\lambda_b)}{q(\lambda_b)}
+
v^{\prime \prime}(\lambda_b) \Big(\log p - \log q(\lambda_b) \Big) + 
\nonumber
\\
&\mbox{~~~~~~}\Big(1-v(\lambda_b) \Big)
\bigg[
\frac{q(\lambda_b) q^{\prime \prime}(\lambda_b) - (q^{\prime}(\lambda_b))^2}{q(\lambda_b)}
\bigg].
\label{Eqn_T1_prime_lambda_b}
\end{align}
Substituting $q(0) = p, q^{\prime}(0) = -p^2, q^{\prime \prime}(0) = 2p^3$ and $\lim_{\lambda_b \rightarrow 0} v(\lambda_b) = 1/2$ in \cref{Eqn_T1_prime_lambda_b} we get
\begin{align}
\lim_{\lambda_b \rightarrow 0} 
T_1^{\prime}(\lambda_b) = 2pv^{\prime}(0) + \frac{p^2}{2},
\label{Eqn_T1_limit}
\end{align}
where $v^{\prime}(0) \coloneqq \lim_{\lambda_b \rightarrow 0} v^{\prime}(\lambda_b)$. 
In the above discussion we showed that, 
$\lim_{\lambda_b \rightarrow 0}F_1(\lambda_b) = p, 
\lim_{\lambda_b \rightarrow 0}F_1^{\prime}(\lambda_b) = -p^2/2$, and $\lim_{\lambda_b \rightarrow 0} T_1(\lambda_b) = -p/2$.
Substituting these limits in \cref{Eqn_F1_prime_limit} and from \cref{Eqn_T1_limit} we get
\begin{align}
\lim_{\lambda_b \rightarrow 0} 
F_1^{\prime \prime}(\lambda_b) = 2p^2v^{\prime}(0) + \frac{p^3}{2} + \frac{p^3}{4}.
\label{Eqn_F1_primeprime_limit}
\end{align}
Using similar steps, $\lim_{\lambda_b \rightarrow 0} 
F_2^{\prime \prime}(\lambda_b)$ can be evaluated. We skip the details due to space constraints. The limit of $F_2^{\prime \prime}(\lambda_b)$ is given by
\begin{align}
\lim_{\lambda_b \rightarrow 0} 
F_2^{\prime \prime}(\lambda_b) = -2p^2v^{\prime}(0) + \frac{p^3(p-2)}{2(1-p)} + \frac{p^4}{4(1-p)}.
\label{Eqn_F2_primeprime_limit}
\end{align}
Simplifying \cref{Eqn_F1_primeprime_limit} and \cref{Eqn_F2_primeprime_limit} we get
\begin{align}
\lim_{\lambda_b \rightarrow 0} 
F^{\prime \prime}(\lambda_b) 
% 
%=  \lim_{\lambda_b \rightarrow 0} F_1^{\prime \prime}(\lambda_b) + F_2^{\prime \prime}(\lambda_b)
% 
= \frac{-p^3}{4(1-p)}
= \frac{-1}{4\lambda_w(\lambda_w+1)^2},
\label{Eqn_F_primeprime_limit}
\end{align}
where in the last step we have substituted the value of $p = 1/(\lambda_w+1)$ (since $\mu = 1$).
Using \cref{Eqn_F_lambda_b_limit}, \cref{Eqn_F_prime_limit}, and \cref{Eqn_F_primeprime_limit} in \cref{Eqn_I_err_derivatives} we get
$I_{err}(0) = 0, I_{err}^{\prime}(0) = 0, $ and $I_{err}^{\prime \prime}(0) = 1/\big(4\lambda_w(\lambda_w+1)^2\big)$.
Substituting this in \cref{Eqn_Taylor_series_approx} we get
\begin{align}
I_{err}(\lambda_b) \approx \frac{\lambda_b^2}{8\lambda_w(\lambda_w+1)^2}.
\label{Eqn_Taylor_series_approx_final}
\end{align}
% 

%---------------------------------------------
~\\
\textit{Part-II: Covertness criteria}

From \cref{Eqn_covertness_criterion_asymptotic}, for $\epsilon$-covertness we need
\begin{align}
I_{err} \leq \frac{1}{N} \log \frac{K(N)}{1 - \epsilon},
\label{Eqn_UB_I_Err}
\end{align}
Substituting \cref{Eqn_Taylor_series_approx_final} in \cref{Eqn_UB_I_Err}, we get
\begin{align}
\lambda_b \leq \sqrt { \frac{8\lambda_w (\lambda_w + 1)^2}{N} \log \frac{K(N)}{1 - \epsilon} }
\label{Eqn_UB_lambda_b_main}
\end{align}
and this proves the required upper bound on $\lambda_b$.
Note that in \cref{Eqn_UB_lambda_b_main}, $\lambda_w$ and $\epsilon$ are constants and $K(N)$ is a function of $N$ that decreases with $N$ at a sub-exponential rate (see \cref{Eqn_const_KF_KM}).
Thus $\lambda_b$ should be order $\mathcal{O}(\sqrt{\log(K(N))/N})$ and this completes the proof of the theorem.
\end{proof}

% \begin{remark}
% \label{Remark_Tayler_closeness}
% %  
% 
% % 
% \end{remark}

%==================================
%==================================
% 
\section{Conclusion}
\label{Section_Conclusion}

Departing from the usual i.i.d. statistics, in this work we consider a covert queueing problem with a Markovian
statistic. We assume that Bob is a bufferless server who
informs Willie about the server states (busy or idle) as seen by N successive arrivals. We formulate this covert queueing
setup as a hypothesis testing problem between two 
Markov Chains $P$ and $Q$ and identify the upper bound
on $\lambda_b$ that ensures covertness in admitting Nillie jobs.

As part of future work, it would be interesting to extend this work for the setting where Bob has a queue for the arriving jobs to wait. In this case, the information metric that Bob could use is the sequence of queue lengths as seen by the arriving customers. Relaxing the service times from exponential to general distributions and considering a multi-server setting of the problem is also for future work.

%==========================================
%==========================================
% 
\section*{Acknowledgments}
This work is supported by the DST-INSPIRE faculty program of Government of India.

%=========================================================

\bibliographystyle{IEEEtran}
\bibliography{References_Covert_communication}

\end{document}